\documentclass[a4paper,11pt]{article}
\usepackage{microtype}
\usepackage[utf8]{inputenc}
\usepackage{amsthm,amsmath,amssymb}
\usepackage[dvipsnames,svgnames]{xcolor}
\usepackage{todonotes}
\usepackage{framed}
\usepackage{graphicx}
\usepackage{url}
\usepackage{fullpage}

\newif\ifdraft\drafttrue

\newtheorem{corollary}{Corollary}
\newtheorem{proposition}{Proposition}

\newcommand{\Z}{\mathbb Z}
\newcommand{\N}{\mathbb N}
\newcommand{\U}{\mathcal U}
\newcommand{\G}{\mathcal G}

\newcommand{\Lin}{\mathcal L}

\newtheorem{theorem}{Theorem}
\newtheorem{fact}{Fact}
\newtheorem{lemma}{Lemma}
\newtheorem{definition}{Definition}

\newcommand{\set}[1]{\left\{ #1 \right\}}

\newcommand{\A}{\mathcal{A}}
\newcommand{\VE}{\text{V}_{\text{Eve}}}
\newcommand{\VA}{\text{V}_{\text{Adam}}}

\newcommand{\vinit}{v_{\text{init}}}
\newcommand{\qinit}{q_{\text{init}}}
\newcommand{\sinit}{s_{\text{init}}}
\newcommand{\last}{\text{last}}
\newcommand{\dist}{\text{dist}}

\newcommand{\MP}[1]{\text{MeanPayoff}_{ #1 }}
\newcommand{\MPn}[2]{\text{MeanPayoff}_{ #2, #1 }}

\title{The complexity of mean payoff games\\ using universal graphs}
\author{Nathana{\"e}l Fijalkow, Pawe{\l} Gawrychowski, Pierre Ohlmann}
\date{\today}

\begin{document}

\maketitle

\begin{abstract}
We study the computational complexity of solving mean payoff games.
This class of games can be seen as an extension of parity games, 
and they have similar complexity status: in both cases solving them is in $\textbf{NP} \cap \textbf{coNP}$ and not known to be in~$\textbf{P}$.
In a breakthrough result Calude, Jain, Khoussainov, Li, and Stephan constructed in 2017 a quasipolynomial time algorithm for solving parity games,
which was quickly followed by two other algorithms with the same complexity.
Our objective is to investigate how these techniques can be extended to the study of mean payoff games.

The starting point is the notion of separating automata introduced by Boja{\'n}czyk and Czerwi{\'n}ski, 
which have been used to present all three quasipolynomial time algorithms for parity games
and gives the best complexity to date.
The notion naturally extends to mean payoff games and yields a class of algorithms for solving mean payoff games.
The contribution of this paper is to prove tight bounds on the complexity of algorithms in this class.

Colcombet and Fijalkow proved that separating automata are equivalent to universal graphs.
The technical results of this paper are to give upper and lower bounds on the size of universal mean payoff graphs.
The upper bounds yield two new algorithms for solving mean payoff games.
Our first algorithm depends on the largest weight $N$ (in absolute value) appearing in the graph 
and runs in sublinear time in $N$,
improving over the previously known linear dependence in $N$.
Our second algorithm runs in polynomial time for a fixed number $k$ of weights.

We complement our upper bounds by providing in both cases almost matching lower bounds on the size of universal mean payoff graphs, 
showing the limitations of the approach. 
We show that we cannot hope to improve on the dependence in $N$ nor break the linear dependence in the exponent in the number $k$ of weights.
In particular, this shows that separating automata do not yield a quasipolynomial algorithm for solving mean payoff games.
\end{abstract}

\section{Introduction}
A mean payoff game is played over a finite graph whose edges are labelled by integer weights.
The interaction of the two players, called Eve and Adam, describe a path in the graph.
The goal of Eve is to ensure that the (infimum) limit of the weights average is non-negative.

\begin{figure}[ht]
\centering
\includegraphics[width=.5\linewidth]{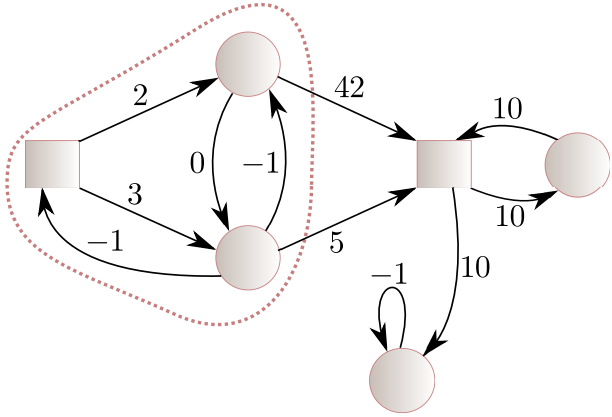}
\caption{A mean payoff game (the dotted region is the set of vertices from which Eve has a strategy ensuring mean payoff).}
\label{fig:mp_game}
\end{figure}

\paragraph*{Early results}
The model of mean payoff games was introduced independently by Ehrenfeucht and Mycielski~\cite{EhrenfeuchtMycielski79}
and by Gurvich, Karzanov, and Khachiyan~\cite{GKK88}.
A fundamental property proved in both papers is that such games are positionally determined,
meaning that for both players, if there exists a strategy ensuring mean payoff, then there exists one using no memory at all.
This holds for both players and is the key argument implying the intriguing complexity status of solving mean payoff games:
the decision problem is in $\textbf{NP}$ and in $\textbf{coNP}$, but not known to be solvable in polynomial time.
It is unlikely to be $\textbf{NP}$-complete, since this would imply that $\textbf{NP} = \textbf{coNP}$.
Hence such a problem is either solvable in polynomial time or an interesting piece in the landscape of computational complexity.

This is one of the reasons that makes the study of mean payoff games exciting.
In addition to game theory, the study of mean payoff games is motivated by verification and synthesis of programs, 
and by their intricate connections to optimisation and linear programming.

\paragraph*{Towards a polynomial time algorithm}
The seminal paper of Zwick and Paterson~\cite{ZwickPaterson96} relates mean payoff games to discounted payoff games
and simple stochastic games, and most relevant to our work, constructs a pseudopolynomial algorithm for solving mean payoff games.
An algorithm is said to be pseudopolynomial if it is polynomial when the weights are given in unary.
More specifically, the algorithm constructed by Zwick and Paterson depends linearly on the largest weight $N$ (in absolute value) 
appearing in the game.
If the weights are given in unary, $N$ is indeed polynomial in the representation.

The question whether there exists a polynomial time algorithm for mean payoff games with the usual representation of weights, meaning in binary,
is open.

\paragraph*{Recent contributions}
The algorithm constructed by Zwick and Paterson runs in time $O(n^3 m N)$,
where $n$ is the number of vertices, $m$ the number of edges, and $N$ the largest weight in the graph.
This has been improved by Brim, Chaloupka, Doyen, Gentilini, and Raskin~\cite{BCDGR11}, yielding an algorithm of complexity $O(n m N)$.
Solving a mean payoff game is very related to constructing an optimal strategy, meaning one achieving the highest possible value.
The state of the art for this problem is due to Comin and Rizzi~\cite{CominRizzi17}, 
who designed a pseudopolynomial time algorithm; the achieved complexity is linear in $N$, as for the best algorithms for solving mean payoff games.

The model of mean payoff games has been recently shown to be strongly connected to the existence of a strongly polynomial time algorithm
for linear programming, which is the 9\textsuperscript{th} item in Smale's list of problems for the 21\textsuperscript{st} century~\cite{Smale98}.
Indeed, Allamigeon, Benchimol, Gaubert, and Joswig~\cite{ABGJ14} have shown that a strongly polynomial time semi-algebraic pivoting rule 
in linear programming would solve mean payoff games in strongly polynomial time.

\paragraph*{Parity games}
It is most instructive in this context to think of parity games as a subclass of mean payoff games. 
More specifically, they are mean payoff games whose weights are of the form $(-n)^p$ for $p \in \set{1,\ldots,d}$,
and $d$ is the number of priorities.
The breakthrough result of Calude, Jain, Khoussainov, Li, and Stephan~\cite{CJKL017} 
was to construct a quasipolynomial time algorithm for solving parity games.
Following decades of exponential and subexponential algorithms, 
this very surprising result triggered further research: 
soon after two different quasipolynomial time algorithms
were constructed, by Jurdzi{\'n}ski and Lazi{\'c}~\cite{JL17}, and by Lehtinen~\cite{Leh18}.
They report almost the same complexity, which is roughly $n^{O(\log d)}$.
%\cite{FJSSW17}

Boja{\'n}czyk and Czerwi{\'n}ski~\cite{BC18} introduced the separation question,
describing a family of algorithms for solving parity games based on separating automata,
and showed that the first quasipolynomial time algorithm yields a quasipolynomial solution to the separation question.
Later, Czerwi{\'n}ski, Daviaud, Fijalkow, Jurdzi{\'n}ski, Lazi{\'c}, and Parys~\cite{CDFJLP18} showed that 
the other two quasipolynomial algorithms also yield quasipolynomial solutions to the separation question,
leading to algorithms with the best complexity to date.
The main contribution of~\cite{CDFJLP18} is a quasipolynomial lower bound on the complexity of algorithms within this framework.
%This paper provides a similar story for mean payoff games. 

Since all quasipolynomial time algorithms for parity games can be presented using separating automata with matching complexity,
this motivates the study of algorithms for solving mean payoff games using separating automata.
Lower bounds for such algorithms provide evidence of the difficulty of extending to mean payoff games the recent ideas for solving parity games.

\paragraph*{Universal graphs}
The notion of universal graphs was developed in the context of parity games by Colcombet and Fijalkow~\cite{CF18}.
The main result is that \emph{separating automata and universal graphs are equivalent}.
More specifically, a separating automaton induces a universal graph of the same size, and conversely.
%This holds for all games which are positionally determined for Eve, so in particular for parity and mean payoff games.
(We refer to Section~\ref{sec:separating_automata} for technical details.)
The point of view of universal graphs is a key step for understanding separating automata\footnote{The paper~\cite{CDFJLP18} uses the notion of universal trees, which came about before universal graphs. Universal trees are specific to the parity condition, unlike universal graphs which can be naturally instantiated to different conditions. The paper~\cite{CF18} shows that universal trees are exactly saturated universal graphs for parity.}.

The notion of universal graphs has been extensively studied in an unrelated context of labelling schemes,
that seek to assign a short bitstring (called a label) to every node of a graph, so that a query concerning
two nodes can be answered by looking at their corresponding labels alone.  
As a prime example, labelling nodes of an undirected graph for adjacency queries is known to be equivalent to constructing
a so-called induced-universal graph~\cite{KannanNR92}. 
Even though for some queries approaches based on an appropriately chosen notion of universal graphs are known to be suboptimal~\cite{FreedmanGNW17},
we also have examples in which they allow for a significantly simpler and more efficient solution~\cite{GawrychowskiKLP18}.

\paragraph*{Contributions}
This paper is concerned with the computational complexity of solving mean payoff games.
The best algorithm to date has complexity $O(nmN)$, where $n$ is the number of vertices, $m$ the number of edges, and $N$ the largest weight.
A fourth parameter is relevant: $k$, the number of distinct weights.

This paper gives upper and lower bounds on the complexity of algorithms solving mean payoff games using separating automata,
and in particular yields an improvement over the previously known algorithms.

We start by defining in Section~\ref{sec:separating_automata} separating automata and showing how they can be used to solve mean payoff games.
We show that to construct algorithms for solving mean payoff games it is enough to construct small separating automata.
Section~\ref{sec:universal_graphs} is devoted to the equivalence between separating automata and universal graphs, 
justifying that we focus in this paper on universal graphs.
We give a first very simple construction yielding an algorithm of complexity $O(nmN)$, 
matching the best algorithm to date.

The appeal and beauty of the universal graph technology is that from this point onwards, we do not talk about games anymore.
Indeed, the rest of the paper proves upper and lower bounds on the size of universal graphs.
Our main results are the following.

\paragraph*{Section~\ref{sec:largest_weight}: Universal graphs parametrised by the largest weight}
\begin{itemize}
	\item There exists an $(n,(-N,N))$-universal graph of size at most 
	$2 \left( nN - ((nN)^{1/n} - 1)^n \right)$, which is bounded by $2n (nN)^{1 - 1/n}$.
	\item All $(n,(-N,N))$-universal graphs have size at least $N^{1 - 1/n}$.
\end{itemize}
Since $n$ is polynomial in the size of the input, we can say that $n$ is ``small'', 
while $N$ is exponential in the size of the input when weights are given in binary, hence ``large''.
The multiplicative gap between upper and lower bound is bounded by $2n^2$, hence small.

As a consequence of the upper bound 
we obtain an algorithm for solving mean payoff games improving over the best algorithm to date:
the dependence in $N$ goes from $N$ to $N^{1 - 1/n}$, yielding the first algorithm running in sublinear time in $N$.
The lower bound shows that this dependence cannot be improved for algorithms using separating automata.

\paragraph*{Section~\ref{sec:number_weights}: Universal graphs parametrised by the number of weights}
\begin{itemize}
	\item For all $W \subseteq \Z$ of size~$k$, there exists an $(n,W)$-universal graph of size $n^k$.
	\item For all $k$, for $n$ large enough, 
	there exists $W \subseteq \Z$ of size~$k$ such that all $(n,W)$-universal graphs have size at least $\Omega(n^{k-2})$.
\end{itemize}
As a consequence, we obtain an algorithm solving mean payoff games of complexity $O(m n^k)$.
The lower bound shows that algorithms using separating automata cannot break the $O(n^{\Omega(k)})$ barrier,
and in particular do not have quasipolynomial complexity.

\section{Definitions}
We write $[i,j]$ for the interval $\set{i,i+1,\dots,j-1,j}$, and use parenthesis to exclude extremal values,
so $[i,j)$ is $\set{i,i+1,\dots,j-1}$.

\paragraph*{Graphs}
A labelled directed graph is given by a set $V$ of vertices, a set $E \subseteq V \times \Z \times V$ of edges,
and an initial vertex $\vinit$, so $(v,w,v') \in E$ means that there is an edge from $v$ to $v'$ labelled by $w \in \Z$.
We let $n$ denote the number of vertices and $m$ the number of edges.
The size of a graph $G$ is its number of vertices and denoted by $|G|$.
For $W \subseteq \Z$ a finite subset of the integers, 
we speak of an $(n,W)$-graph if it has size at most $n$ and all the weights labelling edges are in $W$.
Sometimes we use $W$-graph if the size is irrelevant.

A path $\pi$ is a (finite or infinite) sequence of consecutive triples $(v,w,v')$ in $E$.
Consecutive means that the third component of a triple in the sequence matches the first component of the next triple. 
We write $\pi = (v_0,w_0,v_1) (v_1,w_1,v_2) \cdots$
and let $\pi_{\le i}$ denote the prefix of $\pi$ of length $i$, meaning $\pi_{\le i} = (v_0,w_0,v_1) \cdots (v_{i-1},w_{i-1},v_i)$.
In the case of a finite path we write $\last(\pi)$ for the last vertex in $\pi$,
and talk of a path from $v_0$ to $\last(\pi)$.

We always assume that for all vertices $v$ there exists a path from $\vinit$ to $v$.

\paragraph*{Games}
A mean payoff game is given by a graph together with two sets $\VE$ and $\VA$ such that $V = \VE \uplus \VA$.
We often let $\G$ denote a mean payoff game.
The set $\VE$ is the set of vertices controlled by Eve (represented by circles in Fig~\ref{fig:mp_game})
and $\VA$ the set of vertices controlled by Adam (represented by squares in Fig~\ref{fig:mp_game}).
We speak of an $(n,W)$-mean payoff game if the underlying graph is an $(n,W)$-graph.

The interaction between the two players goes as follows.
A token is initially placed on the initial vertex $\vinit$, and the player who controls this vertex pushes the token along an edge, 
reaching a new vertex; the player who controls this new vertex takes over, 
and this interaction goes on.
To avoid dealing with sinks we will always assume that in a game for all vertices $v$ there exists $(v,w,v') \in E$,
hence the interaction results in an infinite path.
We say that an infinite path $\pi$ satisfies mean payoff if
\[
\liminf_\ell \frac{1}{\ell} \sum_{i = 0}^{\ell - 1} w_i \ge 0.
\]

\paragraph*{Strategies}
A strategy is a map $\sigma : E^* \to E$.
Note that we always take the point of view of Eve, so a strategy implicitely means a strategy of Eve.
We say that a path $\pi$ is consistent with the strategy $\sigma$ if
for all $i$, if $v_i \in \VE$, then 
\[
\sigma(\pi_{\le i}) = (v_i, w_i, v_{i+1}).
\]
A strategy $\sigma$ ensures mean payoff if all paths starting from $\vinit$ and consistent with $\sigma$ satisfy mean payoff.
Solving mean payoff games is the following decision problem:
\begin{framed}
\noindent\begin{tabular}{@{}rl}
\textbf{INPUT}: & a mean payoff game $\G$ \\
\textbf{OUTPUT}: & YES if Eve has a strategy ensuring mean payoff, NO otherwise.
\end{tabular}
\end{framed}

\paragraph*{Positional strategies}
Of special importance are positional strategies, which are given by 
\[
\sigma : V \to E.
\]
Such strategies make decisions only considering the current vertex. 
A positional strategy induces a strategy $\widehat{\sigma} : E^* \to E$
by $\widehat{\sigma}(\pi) = \sigma(\last(\pi))$, where by convention the last vertex of the empty path is $\vinit$.

\begin{theorem}[\cite{EhrenfeuchtMycielski79}]\label{thm:positional}
For all mean payoff games, if Eve has a strategy ensuring mean payoff, then she has a positional strategy ensuring mean payoff.
\end{theorem}

We say that a graph satisfies mean payoff if all infinite paths satisfy mean payoff.
This can be easily characterised using cycles: 
a cycle is a path of finite length $\ell$ such that $v_0 = v_\ell$.
A cycle $(v_0,w_0,v_1) \cdots (v_{\ell-1},w_{\ell-1},v_0)$ is negative if its total weight is negative, meaning
\[
\sum_{i = 0}^{\ell - 1} w_i < 0.
\]

\begin{fact}\label{fact:graph_mp}
A graph satisfies mean payoff if and only if it does not contain any negative cycles.
%\begin{itemize}
%	\item A graph does not contain any negative cycle if and only if all infinite paths in the graph satisfy mean payoff.
%	\item A graph satisfies mean payoff if and only if it does not contain any negative cycles nor any sink.
%\end{itemize}
\end{fact}

Given a mean payoff game $\G$ and a positional strategy $\sigma$, we let $\G[\sigma]$ denote the graph
obtained by restricting $\G$ to the moves prescribed by $\sigma$.
Formally, the set of vertices and edges are
\[
\begin{array}{lll}
V[\sigma] & = & \set{v \in V : \text{ there exists a path consistent with } \sigma \text{ from } \vinit \text{ to } v}\\
E[\sigma] & = & \set{(v,w,v') \in E : v \in \VA \text{ or } \left( v \in \VE \text{ and } \sigma(v) = (v,w,v') \right) }.
\end{array}
\]
%\[
%\begin{array}{lll}
%V[\sigma] & = & \set{v \in V : \begin{array}{c} \text{ there exists a path consistent} \\ \text{with } \sigma \text{ from } \vinit \end{array}}\\
%E[\sigma] & = & \set{(v,w,v') \in E : \begin{array}{c} v \in \VA \text{ or } \\ \left( v \in \VE \text{ and }\right. \\ 
%\left. \sigma(v) = (v,w,v') \right) \end{array} }.
%\end{array}
%\]

\begin{fact}\label{fact:game_mp}
Let $\G$ be a mean payoff game and $\sigma$ a positional strategy.
Then the graph $\G[\sigma]$ satisfies mean payoff if and only if $\sigma$ ensures mean payoff.
\end{fact}

\paragraph*{Safety games}
The general approach we are following in this paper is to reduce the problem of solving mean payoff games to solving 
a much simpler class of games called safety games.

A safety game is played over an unlabelled directed graph, meaning the set of edges is $E \subseteq V \times V$.
Formally, a safety game is given by a graph together with two sets $\VE$ and $\VA$ such that $V = \VE \uplus \VA$.
The definitions of paths and strategies are easily adapted from the case of mean payoff games.
The word safety refers to the winning condition, which is the simplest possible: 
we distinguish a set of vertices $T$ and say that a path satisfies safety if it does not contain any vertex in $T$.
A strategy $\sigma$ ensures safety if all paths starting from $\vinit$ and consistent with $\sigma$ satisfy safety.
Solving a safety game is determining whether Eve has a strategy ensuring safety.
%The positional determinacy stated in Theorem~\ref{thm:positional} also holds for safety games.
The following lemma is folklore.

\begin{lemma}\label{lem:safety_games}
There exists an algorithm running in time $O(m)$ for solving safety games.
\end{lemma}

Throughout the paper we fix $n$ a natural number and $W$ a finite set of integers.

\label{sec:mean_payoff}

\section{Solving mean payoff games using separating automata}
%\begin{figure*}[!ht]
%\centering
%\includegraphics[width=.5\linewidth]{separation.eps}
%\caption{The separation question.}
%\label{fig:separation}
%\end{figure*}

The notion of separating automata was introduced in~\cite{BC18} as a framework for constructing algorithms solving parity games, 
and further studied in~\cite{CDFJLP18}.

\subsection*{Separating automata}
We consider deterministic safety automata over the alphabet $W$:
such an automaton is given by a finite set of states $Q$, a special rejecting state $\bot$, an initial state $\qinit$, 
and a transition function $\delta : Q \times W \to Q$.
For a word $w = w_0 w_1 \cdots \in W^\omega$, the run over $w$ is the word $\rho = q_0 q_1 \cdots \in Q^\omega$
such that $q_0 = \qinit$ and for all $i \in \N$, we have $q_{i+1} = \delta(q_i,w_i)$.
Safety means that a run is accepting if it does not contain $\bot$, and an infinite word $w$ is accepted if the run $\rho$ over~$w$ is accepting.
The language recognised by an automaton is the set of accepted words.
The size of an automaton is its number of states (not counting the rejecting state $\bot$).

We let 
\[
\MP{W} = \set{w : \liminf_\ell \frac{1}{\ell} \sum_{i = 0}^{\ell - 1} w_i \ge 0} \subseteq W^\omega.
\]
Recall that a path in a graph is a sequence of edges, and since edges carry a label, a path naturally induces a sequence of labels.
We say that a path is accepted or rejected by an automaton; this is an abuse of language
since what the automaton reads is only the labels of the corresponding path.
For a graph~$G$ we let $\text{Paths}(G) \subseteq W^\omega$ denote the set of sequences induced by paths in $G$.
We let
\[
\MPn{W}{n} = \bigcup \set{ \text{Paths}(G) : \begin{array}{c} G \text{ is an } (n,W)\text{-graph} \\ \text{satisfying mean payoff}\end{array}}.
\]
\begin{definition}
An automaton is $(n,W)$-\textit{separating} if the language $L$ it recognises satisfies 
\[
\MPn{W}{n} \subseteq L \subseteq \MP{W}.
\]
\end{definition}

\subsection*{Reduction to safety games}

The definition of separating automata was designed for the following lemma to hold.

\begin{lemma}\label{lem:reduction1}
Let $\G$ be a $(n,W)$-mean payoff game and $\A$ an $(n,W)$-separating automaton recognising the language $L$.
Then Eve has a strategy ensuring mean payoff if and only if she has a strategy ensuring $L$.
\end{lemma}
\begin{proof}
Let us first assume that Eve has a strategy $\sigma$ ensuring mean payoff.
It can be chosen positional thanks to Theorem~\ref{thm:positional}.
Then $\G[\sigma]$ satisfies mean payoff thanks to Fact~\ref{fact:game_mp}.
Since $\MPn{W}{n} \subseteq L$ this implies that $\text{Paths}(\G[\sigma]) \subseteq L$.
In other words, all paths consistent with $\sigma$ from $\vinit$ are in $L$, or equivalently $\sigma$ ensures $L$.
Conversely, assume that Eve has a strategy $\sigma$ ensuring $L$.
Since $L \subseteq \MP{W}$, the strategy $\sigma$ also ensures mean payoff.
\end{proof}

It follows that solving $(n,W)$-mean payoff games is equivalent to solving games whose condition is given by a $(n,W)$-separating automaton.
Since separating automata are safety deterministic automata, the latter can naturally be reduced to safety games, as we now explain.

We reduce an $(n,W)$-mean payoff game $\G$ to a safety game on $\G \times \A$.
We let $(V,E)$ denote the underlying graph of $\G$ and write $\A = (Q,\qinit,\delta)$.
In $\G \times \A$ the set of vertices and edges are
\[
\begin{array}{lll}
\VE' & = & \VE \times Q, \\
\VA' & = & \VA \times Q, \\
E'   & = & \set{\left( (v,q),\ (v',\delta(q,w)) \right) : (v,w,v') \in E}.
\end{array}
\]
The initial vertex is $(\vinit,\qinit)$.
To define the safety condition we choose $V \times \set{\bot}$ to be the losing vertices: 
a path satisfies safety if it does not contain any vertex from $V \times \set{\bot}$.

\begin{lemma}\label{lem:reduction2}
Let $\G$ be an $(n,W)$-mean payoff game and $\A$ an $(n,W)$-universal graph.
Then Eve has a strategy in $\G$ ensuring $L$ if and only if 
she has a strategy in $\G \times \A$ ensuring safety.
\end{lemma}

\begin{proof}
We construct a bijection between paths in $\G$ and in $\G \times \A$.
The path $\pi$ in $\G$ 
induces the path $\pi' = ((v_0,q_0),w_0,(v_1,q_1))\ ((v_1,q_1),w_1,(v_2,q_2)) \cdots$ in $\G \times \A$
where $q_i = \delta(\qinit,w_0 w_1 \cdots w_{i-1})$.
This bijection induces a one-to-one correspondence between strategies in $\G$ and in $\G \times \A$.
Since $\pi$ is in $L$ if and only if $\pi'$ satisfies safety, 
strategies in $\G$ ensuring $L$ correspond to strategies in $\G \times \A$ ensuring safety, and vice-versa.
\end{proof}

Combining Lemma~\ref{lem:reduction1} and Lemma~\ref{lem:reduction2} yields an algorithm for solving mean payoff games whose complexity is proportional to the size of $\A$.

\begin{theorem}\label{thm:reduction}
Given an $(n,W)$-separating automaton $\A$, 
we can construct an algorithm solving $(n,W)$-mean payoff games of complexity $O(m \cdot |\A|)$.
\end{theorem}

The algorithm is simply to construct the safety game $\G \times \A$ and then to solve it in linear time using Lemma~\ref{lem:safety_games}.

\label{sec:separating_automata}

\section{Universal graphs}
The technical objective of this paper is to give upper and lower bounds on the size of separating automata.
To gain some insights into the structure of separating automata, we use the equivalent notion of universal graphs.

\subsection*{Definition of universal graphs}
For two graphs $G,G'$, a homomorphism $\phi : G \to G'$
maps the vertices of $G$ to the vertices of $G'$ such that
\[
(v,w,v') \in E \ \implies\ (\phi(v),w,\phi(v')) \in E'.
\]
Homomorphisms do not have any requirement about the initial vertices.

\begin{definition}
A graph is $(n,W)$-\textit{universal} if it satisfies mean payoff and 
any $(n,W)$-graph satisfying mean payoff can be mapped homomorphically into it.
\end{definition}

It is not clear at this point that there exists a universal graph, let alone a small one.
Indeed, the definition creates a tension between ``satisfying mean payoff'', restricting the cycles,
and ``homomorphically map any $(n,W)$-graph satisfying mean payoff'', implying that the graph witnesses varied behaviours.
A simple construction of an $(n,W)$-universal graph is to take the disjoint union of all $(n,W)$-graphs satisfying mean payoff.
Indeed, up to renaming of vertices there are finitely many such graphs, so this yields a very large but finite $(n,W)$-universal graph. 
We will show that there are much smaller ones.

A $W$-linear graph is given by a finite subset $A$ of the integers.
The set of vertices is $A$ and for any $v,v' \in A$ and $w \in W$ there is an edge from $v$ to $v'$ labelled by $w$ if $v' - v \le w$.
We drop $W$ when it is clear from the context, or irrelevant.
Observe that linear graphs satisfy mean payoff, which follows from the fact that they do not contain negative cycles.

Given a graph with no negative cycles, the distance from a vertex $v$ to a vertex $v'$ is the smallest sum of the weights along a path from $v$ to~$v'$.

\begin{lemma}\label{lem:naive}
Let $G$ be a $W$-graph satisfying mean payoff.
We let $\Lin(G)$ define the $W$-linear graph
\[
\set{ \dist(\vinit,v) : v \in V }.
\]
Then $G$ homomorphically maps into $\Lin(G)$.
\end{lemma}

\begin{proof}
We define $\phi$ by mapping a vertex $v$ of $G$ to $\dist(\vinit,v)$.
Since $G$ satisfies mean payoff it does not contain negative cycles, and by assumption all vertices are reachable from $\vinit$, 
so $\phi$ is well defined.
We claim that $\phi$ is a homomorphism from $G$ to $\Lin(G)$, which follows from the triangle inequality: 
if $(v,w,v') \in E$, then
$\dist(\vinit,v') \le \dist(\vinit,v) + w$,
implying $\phi(v') - \phi(v) \le w$.
\end{proof}

This very simple lemma has two important consequences.

\begin{corollary}\label{cor:linearisation}
For every $(n,W)$-universal graph, there exists a linear $(n,W)$-universal graph of the same size.
\end{corollary}
\begin{proof}
Let $\U$ be an $(n,W)$-universal graph, then $\U$ maps homomorphically into the linear graph $\Lin(\U)$,
and the size of $\Lin(\U)$ (meaning the number of vertices) is no larger than the size of $\U$.
This implies that $\Lin(\U)$ is $(n,W)$-universal, since any $(n,W)$-graph $G$ satisfying mean payoff homomorphically maps into $\U$,
which composed with the homomorphism into $\Lin(\U)$ given by Lemma~\ref{lem:naive}
yields a homomorphism from $G$ to $\Lin(\U)$.
\end{proof}

\begin{corollary}\label{cor:naive}
The $(-N,N)$-linear graph $(-nN,nN)$ is $(n,(-N,N))$-universal.
\end{corollary}
\begin{proof}
Let $G$ be an $(n,(-N,N))$-graph satisfying mean payoff. 
We observe that $\Lin(G) \subseteq (-nN,nN)$, 
so Lemma~\ref{lem:naive} implies that $(-nN,nN)$ is $(n,(-N,N))$-universal.
\end{proof}

We state here a simple fact that we will use several times later on about homomorphisms into linear graphs.
\begin{fact}\label{fact:equality}
Let $G$ be a graph and $A$ a linear graph, $\phi : G \to A$ a homomorphism. 

Consider a cycle $(v_0,w_0,v_1) \cdots (v_{\ell-1},w_{\ell-1},v_0)$ in $G$ of total weight $0$.
Then for $i \in [0,\ell)$, we have 
$\phi(v_{i+1}) - \phi(v_i) = w_i$,
where by convention $v_\ell = v_0$.
\end{fact}
\begin{proof}
By definition of the homomorphism using the edges $(v_i,w_i,v_{i+1})$ for $i \in [0,\ell)$
\[
\begin{array}{lll}
\phi(v_1) - \phi(v_0) & \le & w_0 \\
\phi(v_2) - \phi(v_1) & \le & w_1 \\
& \ \vdots & \\
\phi(v_0) - \phi(v_{\ell-1}) & \le & w_{\ell-1}.
\end{array}
\]
Assume towards contradiction that one of these inequalities is strict. 
Summing all of them yields $0$ on both sides (because the cycle has total weight $0$), so this would imply $0 < 0$, a contradiction.
Hence all inequalities are indeed equalities.
\end{proof}

\subsection*{Equivalence between separating automata and universal graphs}
We now prove that separating automata and universal graphs are two sides of the same coin.
The equivalence was proved for parity games in~\cite{CF18}, using universal trees as an intermediate step.
We reformulate it here in the context of mean payoff games.
A closer inspection indeed reveals that the equivalence holds for any winning condition which is positionally determined for Eve.
\begin{theorem}[\cite{CF18}]\label{thm:equivalence}
\hfill
\begin{itemize}
	\item An $(n,W)$-universal graph induces an $(n,W)$-separating automaton of the same size.
	\item An $(n,W)$-separating automaton induces an $(n,W)$-universal graph of the same size.
\end{itemize}
\end{theorem}

\begin{proof}
Let $\U$ be an $(n,W)$-universal graph, thanks to Corollary~\ref{cor:linearisation} we can assume that $\U$ is linear, that is, $\U \subseteq \Z$.
We construct an $(n,W)$-separating automaton $\A$. The set of states is $\U$ with a rejecting state $\bot$.
The initial state $\sinit$ is the largest vertex in $\U$, and the transition function is defined by 
\[
\delta(s,w) = \max \set{s' \in \U : s' - s \leq w},
\]
with the convention that if $\max \varnothing = \bot$.
Let $L \subseteq W^\omega$ be the language recognised by $\A$, we will show that 
$\MPn{W}{n} \subseteq L \subseteq \MP{W}$.

We start with proving $L \subseteq \MP{W}$. Let $w \in W^\omega$ accepted by $\A$, and $\rho$ the accepting run over~$w$.
Recall that $(s,w,s')$ is an edge in $\U$ if $s' - s \le w$, so in particular $(s,w,\delta(s,w))$ is an edge in $\U$.
Hence we can see $\rho$ as an infinite path in $\U$, and since $\U$ satisfies mean payoff, 
we have $w \in \MP{W}$, so $L \subseteq \MP{W}$.

We now prove $\MPn{W}{n} \subseteq L$. 
Let $w \in \MPn{W}{n}$: it is induced by a path 
$\pi = (v_0,w_0,v_1)(v_1,w_1,v_2) \cdots$
in $G$, which is an $(n,W)$-graph satisfying mean payoff.
Since $\U$ is $(n,W)$-universal, there exists a homomorphism $\phi : G \to \U$.
To show that $w$ is accepted by $\A$ we argue that the run $\rho$ over $w$ is accepting.
Let $\rho = q_0 q_1 \cdots$. We prove by induction on $i$ that $q_i \neq \bot$ and $\phi(v_i) \le q_i$.
We have $q_0 = \sinit$, with $\phi(v_0) \le \sinit$ because $\sinit$ is by definition the largest vertex in $\U$.
Assuming $q_i \neq \bot$ and $\phi(v_i) \le q_i$, we obtain that $\phi(v_{i+1}) - q_i \leq \phi(v_{i+1}) - \phi(v_i) \leq w_i$, 
so $q_{i+1} = \delta(q_i, w_i) \neq \bot$ and $\phi(v_{i+1}) \le q_{i+1}$.
Thus $w$ is accepted by $\A$ and we have $\MPn{W}{n} \subseteq L$.

\vskip1em
We turn to the converse implication.
%from an $(n,W)$-separating automaton, to construct an $(n,W)$-universal graph of the same size.
Let $\A = (Q,\qinit,\delta)$ be an $(n,W)$-separating automaton recognising a language $L$. 
Without loss of generality, we assume that all states are reachable from $\qinit$ in $\A$. 
We construct an $(n,W)$-graph $\U$. The set of vertices is $Q \setminus \{\bot\}$,
and the set of edges is 
\[
\set{(q,w,\delta(q,w)) : q \in Q, w \in W, \delta(q,w) \neq \bot}.
\]
%Implicitly the presence of the edge $(q,w,\delta(q,w))$
%is conditional to the fact that $\delta(q,w)$ is defined.

We start with proving that $\U$ satisfies mean payoff.
Indeed, an infinite path in $\U$ can be seen as a run in $\A$. 
Since $L \subseteq \MP{W}$, infinite runs in $\A$ satisfy mean payoff, hence all infinite paths in $\U$ satisfy mean payoff.

We do not claim that $\U$ is $(n,W)$-universal. To obtain an $(n,W)$-universal graph we proceed as follows.
We consider the set of $(n,W \cup \set{0})$-graphs over the same set of vertices as $\U$ and equip this set with the partial order given by 
inclusion of edges.
We now look at the subset of such graphs which contain $\U$ and satisfy mean payoff, and choose (arbitrarily) a maximal element $\widehat{\U}$ 
in this set.
We argue that $\widehat{\U}$ is $(n,W)$-universal.
First by definition it satisfies mean payoff, and equivalently it does not contain negative cycles.

In order to prove that it any $(n,W)$-graph satisfying mean payoff homomorphically maps into $\widehat{\U}$,
we first define a total preorder on $\widehat{\U}$.
Define $q \leq q'$ if there is an edge $(q,0,q')$ in $\widehat{\U}$, we show that $\le$ is a total preorder.
It is reflexive because adding the edge $(q,0,q)$ to $\widehat{\U}$ would not create negative cycles, 
hence it is already in $\widehat{\U}$ by maximality.
It is transitive: if $q \le q'$ and $q' \le q''$, meaning that the two edges $(q,0,q')$ and $(q',0,q'')$ are in $\widehat{\U}$,
then adding the edge $(q,0,q'')$ to $\widehat{\U}$ would not create negative cycles, hence it is already in $\widehat{\U}$ by maximality,
implying that $q \le q''$.
It is total: assume towards contradiction that neither $q \le q'$ nor $q' \le q$.
The fact that $q \le q'$ does not hold means because of the maximality of $\widehat{\U}$ 
that adding the edge $(q,0,q')$ to $\widehat{\U}$ would create a negative cycle,
implying that there exists a path $\pi$ from $q'$ to $q$ of total negative weight.
For the same reason, there exists a path $\pi'$ from $q$ to $q'$ of total negative weight.
Considering the concatenation $\pi \cdot \pi'$ of the two paths yields a negative cycle in $\widehat{\U}$, a contradiction.
Hence either $q \le q'$ or $q' \le q$, and $\le$ is total.

Let $G$ be an $(n,W)$-graph satisfying mean payoff, we show that $G$ maps homomorphically into $\widehat{\U}$.
We extend the transition function $\delta$ into $\delta^* : Q \times W^* \to Q$, $W^*$ being the set of finite words on alphabet $W$, by $\delta^*(q,\varepsilon) = q$ and 
$\delta^*(q,\pi \cdot w) = \delta(\delta^*(q,\pi),w)$. 
Define $\phi : G \to \widehat{\U}$ by 
\[
\phi(v) = \max \set{\delta^*(\qinit,\pi) : \pi \text{ path from } \vinit \text{ to } v \text{ in } G},
\]
where the maximum is with respect to the total preorder $\le$.
We first note that for any path $\pi$ from $\vinit$ to $v$ in $G$,
since $\MPn{W}{n} \subseteq L$ this path is accepted, \textit{i.e.} 
$\delta^*(\qinit,\pi) \neq \bot$.
Further, for any vertex $v$ by assumption 
there exists a path $\pi$ from $\vinit$ to $v$ in $G$, 
implying that $\phi(v) \neq \bot$.
We show that $\phi$ is a homomorphism. 
Let $(v,w,v') \in E$. There exists a path $\pi$ from $\vinit$ to $v$ in $G$ such that $\phi(v) = \delta^*(\qinit,\pi)$. 
The path $\pi \cdot (v,w,v')$ goes from $\vinit$ to $v'$ in $G$, so 
\[
\phi(v') \ge \delta^*(\qinit,\pi \cdot w) = \delta(\delta^*(\qinit,\pi),w) = \delta(\phi(v), w).
\] 
This implies that $(\delta(\phi(v), w), 0, \phi(v'))$ is an edge 
in $\widehat{\U}$.
Since $(\phi(v), w, \delta(\phi(v), w))$ is also an edge in $\widehat{\U}$,
so is the edge $(\phi(v), w, \phi(v'))$ obtained by transitivity: adding it would not create a negative cycle,
hence it is already in $\widehat{\U}$ by maximality.
Thus $\phi$ is a homomorphism, and $\widehat{\U}$ is an $(n,W)$-universal graph.
\end{proof}

Theorem~\ref{thm:equivalence} shows how the universal graph
constructed in Corollary~\ref{cor:naive} induces a separating automaton of size $O(nN)$.
Combined with Theorem~\ref{thm:reduction}, 
this simple construction already yields an algorithm matching 
the best complexity so far.

\begin{corollary}
There exists an algorithm for solving mean payoff games
of complexity $O(nmN)$.
\end{corollary} 

The separating automaton constructed this way is indeed very simple: 
the set of states is $(-nN,nN)$, the initial state is $nN$, and the transition function is defined by
$\delta(q,w) = \min(nN, q + w)$.
It is an instructive exercise to directly prove that it is indeed an $(n,(-N,N))$-separating automaton.

\label{sec:universal_graphs}

\section{Constructing universal graphs}
The remainder of the paper gives upper and lower bounds on the size of universal graphs for mean payoff games.

Let us recall the results for parity games, which form the class of mean payoff games whose set of weights are of the form
$W = \set{(-n)^p : p \in [1,d]}$.

\begin{theorem}[\cite{CDFJLP18,CF18}]
For all $n,d$, let $W = \set{(-n)^p : p \in [1,d]}$.
\begin{itemize}
	\item There exists an $(n,W)$-universal graph of size $n^{O(\log(d))}$.
	\item All $(n,W)$-universal graphs have size at least $\Omega(n^{\log(d)})$.
\end{itemize}
\end{theorem}
There are three constructions for the upper bound, one for each quasipolynomial time algorithm:
the algorithm constructed by Calude, Jain, Khoussainov, Li, and Stephan~\cite{CJKL017},
the algorithm constructed by Jurdzi{\'n}ski and Lazi{\'c}~\cite{JL17}, 
and the algorithm constructed by Lehtinen~\cite{Leh18}.
We note that the complexity reported in their analysis is not made worse by rephrasing the algorithms using universal graphs, 
and even in some cases very slightly improved.

The technical core of this paper is to extend this study to arbitrary sets of weights $W$,
inducing algorithms for solving mean payoff games.
We consider two parameters on $W$: the largest weight $N$ is absolute value in Section~\ref{sec:largest_weight}, 
in other words the case where $W = (-N,N)$,
and the number of weights, \textit{i.e.} the size of $W$, in Section~\ref{sec:number_weights}.

\label{sec:constructions}

\section{Parametrised by the largest weight}
In this section we focus on the largest weight of $W$ in absolute value as parameter, so we fix $W = (-N,N)$.

We already explained how to construct an $(n,(-N,N))$-universal graph of size $O(nN)$, 
yielding an algorithm with the best known complexity.
In this section we show the following improved results.

\begin{theorem}
For all $n,N$,
\begin{itemize}
	\item There exists\footnote{For simplicity we assume that $(nN)^{1/n}$ is an integer, 
	otherwise it should be understood as $\lceil (nN)^{1/n} \rceil$, which increases the upper bound by a factor of $2$.} 
	an $(n,(-N,N))$-universal graph of size at most $2 \left( nN - ((nN)^{1/n} - 1)^n \right)$,
	which is bounded by $2n (nN)^{1 - 1/n}$.
	\item All $(n,(-N,N))$-universal graphs have size at least $N^{1 - 1/n}$.
\end{itemize}
\end{theorem}

%In stating both bounds we omit rounding numbers, in other words we assume that $(nN)^{1/n}$ is an integer.

\begin{corollary}
There exists an algorithm for solving mean payoff games
of complexity 
\[
O(mn (nN)^{1 - 1/n}).
\]
\end{corollary} 

\subsection*{Upper bound}

\begin{proposition}
There exists an $(n,(-N,N))$-universal graph of size $2 \left( nN - ((nN)^{1/n} - 1)^n \right)$.
\end{proposition}

Before giving the formal construction, we give some intuitions.
The simple construction in Lemma~\ref{lem:naive} shows that the linear graph $(-nN,nN)$ is $(n,(-N,N))$-universal.
In this construction the initial vertex $\vinit$ is always mapped to $0$ by the homomorphism.
By allowing ourselves to map $\vinit$ anywhere in the linear graph, we get some slack which enables us to remove some values from $(-nN,nN)$ 
while staying universal.

\begin{proof}
Let $b = (nN)^{1/n}$.
We write integers $a \in [0,2nN)$ in basis $b$, hence using $n+1$ digits written $a[i] \in [0,b)$, that is,
\[
a = \sum_{i = 0}^n a[i] (nN)^{i/n}.
\]
Note that since $a \in [0,2nN)$ the $(n+1)$\textsuperscript{th} digit is either $0$ or $1$.
We let $B$ be the set of integers in $[0,2nN)$ which have at least one zero digit among the first $n$ digits in this decomposition.
We see $B$ as a $(-N,N)$-linear graph and argue that it is an $(n,(-N,N))$-universal graph.

Let $G$ be an $(n,(-N,N))$-graph satisfying mean payoff.
Thanks to Lemma~\ref{lem:naive} 
the graph $G$ homomorphically maps into the $(-N,N)$-linear graph $\Lin(G)$.
We set $\Lin(G) = \set{v_0 < \dots < v_{n-1}}$.
Let $w_i = v_i - v_{i-1}$ for $i \in [1,n-1]$. 
We argue that $w_i \in [1,N)$. This follows from the observation that for any $v \in \Lin(G)$,
either $v = \vinit$ or there exists $v' \in \Lin(G)$ such that $v \neq v'$ and $|v' - v| < N$.
Indeed, recall that $\Lin(G) = \set{ \dist(\vinit,v) : v \in V }$.
We proceed by induction on the number of edges of a shortest path witnessing $\dist(\vinit,v)$.
For $v \in V$, either $v = \vinit$ or $\dist(\vinit,v) = \dist(\vinit,v') + w$ with $(v',w,v) \in~E$,
so $w \in (-N,N)$. In case $w = 0$ we rely on the inductive hypothesis to conclude.

We construct a homomorphism $\phi$ from $\Lin(G)$ to $B$, which composed with the homomorphism from $G$ to $\Lin(G)$
yields a homomorphism from $G$ to $B$.
We note that $\phi$ is actually fully determined by $\phi(v_0)$: 
indeed, let $i \in [1,n-1]$, we have $\phi(v_i) = \phi(v_{i-1}) + w_i$ thanks to Fact~\ref{fact:equality}
applied to the cycle $(v_{i-1},w_i,v_i)(v_i,-w_i,v_{i-1})$.
It follows that 
$\phi(v_i) = \phi(v_0) + \sum_{j = 1}^i w_j$.

We choose $\phi(v_0)$ of the form $\sum_{i = 0}^{n-1} a_i b^i$ for $a_i \in [0,b)$,
\textit{i.e.} the $(n+1)$\textsuperscript{th} digit is $0$, or equivalently $\phi(v_0) < nN$.
Since each $w_j$ is in $[1,N)$, for $i \in [1,n-1]$ we have $\sum_{j = 1}^i w_j \in [1,nN)$,
so $\phi(v_i) \in [0,2nN)$. 
Hence to ensure $\phi(v_0),\phi(v_1),\dots, \phi(v_{n-1}) \in B$ it is enough that each of them
has a zero digit among the first $n$ digits in the decomposition in base~$b$.

We show how to choose $a_0,a_1, \dots, a_{n-1}$ in order to ensure $\phi(v_0),\phi(v_1),\dots, \phi(v_{n-1}) \in B$. 
More precisely, we show by induction on $k \in [0,n)$ that there exist $a_0,\dots, a_k \in [0,b)$ such that 
for any $a_{k+1}, \dots, a_{n-1} \in [0,b)$, for all $i \in [0,k]$, 
the $i$-th digit $\phi(v_i)[i]$ of $\phi(v_i)$ is $0$. 
For $k = 0$, we let $a_0 = \phi(v_0)[0] = 0$, independent of the values of $a_1,\dots, a_{n-1}$.

Let $a_0,\dots, a_{k-1}$ be such that for any choice of $a_k, \dots, a_{n-1}$, 
for any $i \in [0,k)$, we have $\phi(v_i)[i] = 0$. 
Let $a_k \in [0,b)$ be the unique value such that 
$\left(\sum_{i = 0}^k a_i b^i + \sum_{i = 1}^k w_i  \right)[k] = 0$. 
Let $a_{k+1}, \dots, a_{n-1} \in [0,b)$. 
By induction hypothesis, for any $i \in [0,k)$, we have $\phi(v_i)[i] = 0$. 
Now
\[
\phi(v_k)[k] = \left(\phi(v_0) + \sum_{i = 1}^k w_i \right)[k]
 = \left(\sum_{i = 0}^{n-1} a_i b^i + \sum_{i = 1}^k w_i \right)[k]
 = \left(\sum_{i = 0}^k a_i b^i + \sum_{i = 1}^k w_i \right)[k] = 0
\]
%\[
%\begin{aligned}
%\phi(v_k)[k]
%& = \left(\phi(v_0) + \sum_{i = 1}^k w_i \right)[k] \\ 
%& = \left(\sum_{i = 0}^{n-1} a_i b^i + \sum_{i = 1}^k w_i \right)[k] \\ 
%& = \left(\sum_{i = 0}^k a_i b^i + \sum_{i = 1}^k w_i \right)[k] = 0,
%\end{aligned}
%\]
which concludes the induction, and the construction of the homomorphism $\phi$.
The size of $B$ is $2 \left( nN - ((nN)^{1/n} - 1)^n \right)$.
\end{proof}

\subsection*{Lower bound}

\begin{proposition}
Any $(n,(-N,N))$-universal graph has size at least $N^{1 - 1/n}$.
\end{proposition}

\begin{proof}
Let $\U$ be an $(n,(-N,N))$-universal graph. Thanks to Corollary~\ref{cor:linearisation} we can assume that $\U$ is linear.
We construct an injective function
$f : [0,N)^{n-1} \to \U^n$.
For $(w_1,\dots,w_{n-1}) \in [0,N)^{n-1}$, we consider the linear graph
$A = \set{0,w_1,w_1 + w_2,\dots,\sum_{i = 1}^{n-1} w_i}$,
which has size $n$, hence homomorphically maps into~$\U$. 
Let $\phi : A \to \U$.
Define $f(w_1,\dots,w_{n-1})$ to be
\[
\left( \phi(0),\phi(w_1),\phi(w_1 + w_2),\dots,\phi \left( \sum_{i = 1}^{n-1} w_i \right) \right).
\]

To see that $f$ is injective, we apply Fact~\ref{fact:equality} to each cycle
\[
\left( \sum_{i = 1}^{j-1} w_i,w_j,\sum_{i = 1}^j w_i \right) \left( \sum_{i = 1}^j w_i,-w_j,\sum_{i = 1}^{j-1} w_i \right)
\]
for $j \in [1,n-1)$. We obtain
\[
\begin{array}{lll}
\phi(w_1) - \phi(0) & = & w_1 \\
\phi(w_1 + w_2) - \phi(w_1) & = & w_2 \\
& \ \vdots & \\
\phi \left( \sum_{i = 1}^{n-1} w_i \right) - \phi \left(\sum_{i = 1}^{n-2} w_i \right) & = & w_{n-1}.
\end{array}
\]
Since $f$ is injective this implies that 
$N^{n-1} \le |\U|^n$, so $|\U| \ge N^{1 - 1/n}$.
\end{proof}

\label{sec:largest_weight}

\section{Parametrised by the number of weights}
In this section we focus on the size of $W$ as a parameter.

\begin{theorem}
For all $k$,
\begin{itemize}
	\item For all $n$, for all $W \subseteq \Z$ of size $k$, there exists an $(n,W)$-universal graph of size $n^k$.
	\item For $n$ large enough, there exists $W \subseteq \Z$ of size $k$ such that 
	all $(n,W)$-universal graphs have size at least $\Omega(n^{k-2})$.
\end{itemize}
\end{theorem}

\begin{corollary}
There exists an algorithm for solving mean payoff games with $k$ weights
of complexity $O(m \cdot n^k)$.
\end{corollary} 

\subsection*{Upper bound}

Define $||W||_n$ to be the number of different sums of $n$ terms of $W$.

\begin{proposition}
There exists an $(n,W)$-universal graph of size $||W||_n$.
\end{proposition}

\begin{proof}
We simply observe that the linear graph constructed in Lemma~\ref{lem:naive}
is included in the set of different sums of $n$ terms of $W$.
\end{proof}

It follows that there exists an algorithm for solving mean payoff games of complexity $O(m \cdot ||W||_n)$.
In particular for $|W| = k$ this yields an algorithm in $O(m \cdot n^k)$,
which is polynomial for a constant $k$.

\subsection*{Lower bound}
We let 
$T = 1 + n + n^2 + \dots + n^{k-2}$
and 
$W = \set{1,n,n^2,\dots,n^{k-2}, - \frac{n-1}{k-1} T}$.
Note that $W$ has indeed size $k$.

\begin{proposition}
Let $\U$ be an $(n,W)$-universal graph. 
Then 
$|\U| \geq \left( \frac{n-1}{(k-1)^2} \right)^{\frac{(k-1)^2}{k}}$.
\end{proposition}

\begin{proof}
Let $\U$ be an $(n,W)$-universal graph. Thanks to Corollary~\ref{cor:linearisation} we can assume that $\U$ is linear.

We consider a class of $(n,W)$-graphs which are cycles of length $n$, and later use a subset of those for the lower bound.
Let $(w_1,\dots,w_{n-1}) \in \set{1,n,\dots,n^{k-2}}$.
The vertices are $[0,n)$. 
There is an edge $(i-1,w_i,i)$ for $i \in [1,n)$ and an edge $(n-1, -\frac{n-1}{k-1} T, 0)$.
To make the total weight in the cycle equal to $0$,
we assume that each $n^j$ appears exactly $\frac{n-1}{k-1}$ many times in $w_1,\dots,w_{n-1}$.

The cycles we use for the lower bound are described in the following way.
We let $S$ be the set of sequences of $k$ integers in $[0,n)$ such that $\sum_{\ell = 1}^k s_\ell = \frac{n-1}{k-1}$,
where we use the notation $(s_1, s_2, \dots, s_k)$ for an element~$s \in S$.
A tuple of $k-1$ sequences in $S$ induces an $(n,W)$-graph $G$.
Let $(s^{(0)}, \dots, s^{(k-2)}) \in S^{k-1}$, the induced graph is partitioned into $k$ parts.
In the $i$\textsuperscript{th} part the weight $n^j$ is used exactly $s^{(j)}_i$ many times.

\begin{figure}[ht]
\centering
\includegraphics[width=.7\linewidth]{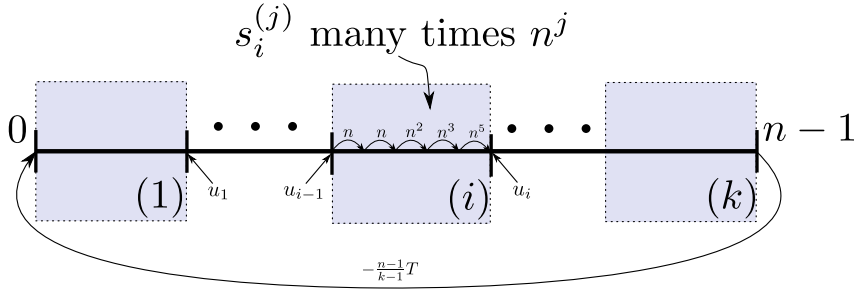}
\label{fig:construction}
\caption{Construction of the graph $G$ from the sequences $(s^{(0)}, \dots, s^{(k-2)}) \in S^{k-1}$.}
\end{figure}
In the drawing the $k$ parts are represented by boxes and numbered from $1$ to $k$ with a number in parenthesis.
To induce the box $(i)$ in this drawing we have 
$s^{(i)}_0 = 0, s^{(i)}_1 = 2, s^{(i)}_2 = 1, s^{(i)}_3 = 1, s^{(i)}_4 = 0$, and $s^{(i)}_5 = 1$.

The vertex in $G$ marking the end of the first box is $u_1 = \sum_{j = 0}^{k-2} s_1^{(j)} n^j$,
and the vertex marking the end of the $i$\textsuperscript{th} box is 
\[
u_i = \sum_{j = 0}^{k-2} \left( \sum_{\ell = 1}^i s_\ell^{(j)} \right) n^j.
\]
We note that writing the number $u_i$ in base $n$ we recover $\sum_{\ell = 1}^i s_\ell^{(j)}$
since this is a number in $[0,\frac{n-1}{k-1}]$, hence in particular in $[0,n)$.
Hence doing this for $u_1,\dots,u_{k-1}$ fully determines the sequences $(s^{(0)}, \dots, s^{(k-2)})$.
We also let $u_0 = 0$.

The constraint on the sums of sequences on $S$ ensures that indeed the graph $G$ is a cycle of total weight~$0$. Hence, $G$ is an $(n,W)$-graph,  and there exists a homomorphism $\phi : G \to \U$.
We define an injective function $f : S^{k-1} \to \U^k$ by
$f(s^{(0)}, \dots, s^{(k-2)}) = \left( \phi \left( u_i \right) : i \in [0,k) \right)$.

To see that $f$ is injective, we note that thanks to Fact~\ref{fact:equality} applied to the cycle $G$, we have
\[
\begin{array}{lll}
\phi(u_1) - \phi(u_0) & = & u_1 \\
\phi(u_2) - \phi(u_1) & = & u_2 \\
& \ \vdots & \\
\phi(u_{k-1}) - \phi(u_{k-2}) & = & u_{k-1},
\end{array}
\]
and as explained above the numbers $u_1,\dots,u_{k-1}$ fully determine the sequences $(s^{(0)}, \dots, s^{(k-2)})$.

Injectivity of $f$ implies that $|S|^{k-1} \le |\U|^k$.
The size of $S$ is 
\[
|S| = \binom{\frac{n-1}{k-1} + k-1}{k-1}
\ge \left( \frac{\frac{n-1}{k-1} + k-1}{k-1} \right)^{k-1} 
\ge \left( \frac{n-1}{(k-1)^2} \right)^{k-1},
\]
%\[
%\begin{array}{ccc}
%|S| & = & \binom{\frac{n-1}{k-1} + k-1}{k-1} \\
%& \ge & \left( \frac{\frac{n-1}{k-1} + k-1}{k-1} \right)^{k-1} \\
%& \ge & \left( \frac{n-1}{(k-1)^2} \right)^{k-1},
%\end{array}
%\]
which implies (for $k$ constant)
$|\U| \geq \Omega \left( n^{ \frac{(k-1)^2}{k} } \right) = \Omega \left( n^{k-2} \right)$.
\end{proof}

\label{sec:number_weights}

\section*{Conclusions}
In this paper we have shown how to extend to mean payoff games 
the ideas developed for constructing quasipolynomial algorithms for parity games.
This relies on the combinatorial notion of separating automata (or equivalently: universal graphs).
We give upper bounds, yielding two new algorithms with the best complexity to date: sublinear in $N$, the largest weight (in absolute value),
and polynomial for a fixed number $k$ of weights.
We provide almost matching lower bounds, showing the limitations of this approach.
In particular, algorithms based on separating automata cannot solve mean payoff games in quasipolynomial time.

More precisely, our lower bounds show that for pathological sets of weights universal graphs are very large.
A more positive note is to consider $W = \set{(-n)^p : p \in [1,d]}$, the set of weights corresponding to parity games:
in this case we know that there exist $(n,W)$-universal graphs of quasipolynomial size (specifically $n^{O(\log(d))}$).
This motivates a deeper understanding of the size of $(n,W)$-universal graphs: for which sets of weights $W$
do there exist small universal graphs? Is there a meaningful hierarchy between parity and mean payoff games?

%Another interesting direction is to understand what are the strategies constructed by our algorithms.
%Indeed, in case Eve has a winning strategy our algorithms synthesise such a strategy.
%It is however not clear whether this strategy is optimal, meaning whether it achieves the highest possible mean payoff value: 
%a priori, it only ensures a non-negative payoff.
%We leave as an open problem to construct algorithms synthesising optimal strategies running in time sublinear in~$N$.

\bibliographystyle{alpha}
\bibliography{bib}

\end{document}